\documentclass[letterpaper,twocolumn,10pt]{article}
\usepackage{usenix2019_v3}

\usepackage{tikz}
\usepackage{amsmath}
\usepackage{amssymb}
\usepackage{xspace}
\usepackage{listings}
\usepackage{algorithm}
\usepackage{algpseudocode}
\usepackage{subcaption}
\usepackage[T1]{fontenc}
\usepackage{inconsolata}
\usepackage{booktabs}
\usepackage{multirow}
\usepackage{comment}
\usepackage{amsthm}
\newtheorem{theorem}{Theorem}[section]

\newtheorem{lemma}[theorem]{Lemma}

\newcommand{\adasum}{Adasum\xspace}

\newcommand{\horovodsum}{sum\xspace}
\newcommand{\Horovodsum}{Sum\xspace}
\newcommand{\norm}[1]{\left\lVert#1\right\rVert}
\newcommand{\mdoubleplus}{\mathbin{{+}\mspace{-6mu}{+}}}
\newcommand{\twodots}{\mathinner {\ldotp \ldotp}}
\newcommand{\ggt}{GGT}
\newcommand{\resnet}{ResNet-50\xspace}
\newcommand{\bertlarge}{BERT-Large\xspace}
\newcommand{\etal}{\textit{et al}.}
\newcommand{\squad}{SQuAD\xspace}
\newcommand{\tf}{TensorFlow\xspace}
\newcommand{\torch}{PyTorch\xspace}

\begin{document}

\date{}

\title{\Large \bf Scaling Distributed Training with Adaptive Summation}

\author{
{\rm Saeed Maleki, Madan Musuvathi, Todd Mytkowicz, Olli Saarikivi}\\
Microsoft Research
\and
{\rm Tianju Xu, Vadim Eksarevskiy, Jaliya Ekanayake, Emad Barsoum}\\
Microsoft\\
\{saemal, madanm, toddm, olsaarik, tix,vaeksare,jaliyaek,ebarsoum\}@microsoft.com
} 

\maketitle

\begin{abstract}
  Stochastic gradient descent (SGD) is an inherently sequential
  training algorithm--computing the gradient at batch $i$ depends on
  the model parameters learned from batch $i-1$.  Prior approaches
  that break this dependence do not honor them (e.g., sum the
  gradients for each batch, which is not what sequential SGD would do)
  and thus potentially suffer from poor convergence. This paper
  introduces a novel method to {\emph combine} gradients called
  \adasum~(for adaptive sum) that converges faster than prior work.
  \adasum is easy to implement, almost as efficient as simply summing
  gradients, and is integrated into the open-source toolkit Horovod.

  This paper first provides a formal justification for \adasum and
  then empirically demonstrates \adasum is more accurate than prior
  gradient accumulation methods.  It then introduces a series of
  case-studies to show \adasum works with multiple frameworks,
  (\tf and \torch), scales multiple optimizers (Momentum-SGD,
  Adam, and LAMB) to larger batch-sizes while still giving good
  downstream accuracy. Finally, it proves that \adasum{} converges.

  To summarize, \adasum scales Momentum-SGD on the MLPerf Resnet50
  benchmark to 64K examples before communication (no MLPerf entry
  converged with more than 16K), the Adam optimizer to 64K examples
  before communication on BERT-LARGE (prior work showed Adam stopped
  scaling at 16K), and the LAMB optimizer to 128K before communication
  on BERT-LARGE (prior work used 64K), all while maintaining
  downstream accuracy metrics.  Finally, if a user does not need to
  scale, we show LAMB with \adasum on BERT-LARGE converges in 30\%
  fewer steps than the baseline.
\end{abstract}

\section{Introduction}Recent trends in deep learning demonstrate that increasing model size, coupled with an increase in training data, results in improved model performance. 
This has led to progressively larger models, such as BERT~\cite{bertorg}, GPT-2~\cite{gpt2}, Megatron~\cite{megatron}, and UniLM~\cite{unilm}. 
This trend along with the end of Moore's law means that these large models require massively parallel architectures to train. 
An important source of parallelism in training is data parallelism where individual nodes train on a subset of data and periodically exchange model updates.
Unfortunately, this is at odds with the sequential nature of stochastic gradient descent (SGD) which is the most common algorithm to train them.
Some prior approaches break this sequential dependence with asynchronous SGD~\cite{Dean2012,hogwild} where individual nodes asynchronously update a global model ignoring potential staleness of model updates. 
Recent advances in hardware with powerful compute nodes with fast interconnects~\cite{dgx1,dgx2,tpulatest} have led to synchronous SGD where one trains with very large minibatch sizes. 
Neither approach is a panacea as both staleness and naively increasing minibatch sizes reduces model convergence~\cite{Keskar2016,Hoffer2017,directsum}

This paper proposes a new approach to data parallelism based on two key insights.  
First, rather than asynchronously updating a global model or increasing the minibatch size, this approach attempts to {\em emulate} a sequential execution in parallel. 
The basic idea is to combine the individual model updates from nodes, each obtained by running a (small) minibatch from a starting model, into an update that would have resulted had these nodes run one after the other from the same starting model. 
Second, this sequential emulation allows us to \emph{sample multiple paths} simultaneously. 
Intuitively, SGD is a stochastic process with each path representing a sample of the possible outcomes. 
By sampling many paths we dramatically reduce the variance of the estimate.  

This paper shows (in Section~\ref{sec:background}) that the following combiner achieves the two properties above:
$$ \adasum(g_1, g_2) = (1 - \frac{g_1^T \cdot g_2}{2\cdot \|g_1\|^2}) g_1 + (1 - \frac{g_1^T \cdot g_2}{2\cdot \|g_2\|^2})g_2 $$
Here $g_1$ and $g_2$ are gradients from individual minibatches, $g_1^T \cdot g_2$ is their dot product, and $\|g\|^2$ represents the norm of the vector $g$. 
This combiner when recursively applied on gradients from all nodes generates a final gradient which can be used to update the starting model with an appropriate learning rate. 
We call this approach \adasum as the combiner represents an adaptive sum of the two gradients with the gradients scaled by an appropriate constant. 


\adasum achieves significant algorithmic efficiency with large batch sizes when compared to synchronous SGD by requiring far less number of epochs to converge to the same loss or model performance. 
This remains true even when using various learning-rate {\em optimizers}, such as Momentum-SGD~\cite{momentum}, Adam~\cite{adam}, and LAMB~\cite{lamb}. 
Alternately, one can use the improved algorithmic efficiency to scale to a much larger effective batch size. 
For example, for Resnet50 \adasum enables Momentum SGD optimizer to converge with an effective batch size of 64K, which is four times larger than the largest batch size we have seen reported for Resnet50 as per MLPerf v0.5 submissions. 
Similarly, for BERT-Large, \adasum enables the Adam optimizer to scale to an effective batch size of 128K. 
Lack of convergence of Adam beyond 16K was the motivation for more sophisticated optimizers such as LARS and LAMB. 
When combined with LAMB, which is the state of the art optimizer for BERT-Large, \adasum converges in 20\% fewer epochs than with LAMB alone when using an effective batch size of 64K. 
In addition, we demonstrate that LAMB with \adasum can also scale to 128K effective batch size. 
These results indicate that \adasum emulates the behavior of much smaller batch size even when running with large batch sizes.

A desirable property of the \adasum operation, as evident from the equation above, is that it has no hyperparameters. 
In all our experiments, we simply reused the recommended hyper-parameters for the baseline synchronous SGD. 
The only additional tuning \adasum entails is a search for suitable base learning rate. 
 
In summary, the contributions of this paper are:
\begin{itemize}
\item \adasum, a new way to combine gradients that scales synchronous SGD to
unprecedented batch sizes, and a proof of its convergence.
\item A detailed discussion of how \adasum is implemented in Horovod, a popular distributed
training framework for \torch and \tf.
\item An evaluation that demonstrates \adasum scales existing
optimizers well beyond what was possible in prior work.  For example,
we demonstrate Adam can scale to 64K examples per allreduce on BERT (16K
before), LAMB to 128K examples per allreduce on BERT (64K before), and
Momentum-SGD to 64K examples per allreduce on ResNet-50 (16K before).  All while
maintaining downstream accuracy and with little hyper-parameter
tuning.
\item An evaluation that demonstrates for similar effective batch
sizes as in prior work, \adasum converges faster than that prior work.
For example, LAMB with \adasum on 64K examples per allreduce on BERT
converges in 30\% fewer steps than LAMB when just averaging gradients.
\end{itemize}



\section{Background}\label{sec:background}
This section provides the background for the the paper, introducing notation and concept used throughout. 

\subsection{Stochastic Gradient Descent}
Machine learning involves learning a model parameterized by a set of weights $w$ based on some training data. 
Given a set of training examples $\{(x_1,y_1),\dots,(x_k,y_k)\}$ where each $x_i$ is a vector representing the input
instance $i$ and $y_i$ is its label, training involves finding a $w$ that minimizes some loss function $L = \sum L_i(w, x_i, y_i)$, 
the sum of individual loss function $L_i$ of the model $w$ on input $(x_i, y_i)$. 

Most training uses stochastic gradient descent~(SGD). SGD starts from an appropriately initialized model $w_0$ and progressively updates the model at step $i$ as
$w_{i+1} = w_i - \alpha_i g_i$. Here $\alpha_i$ is the learning rate at this step as determined by a learning rate schedule, 
and $g_i=\frac{1}{b}\sum_{j=1}^b\nabla L_{j}(w_i)$ is the sum of gradients of individual loss functions for a randomly chosen {\em minibatch} of data of size $b$. 
The stochasticity of SGD arises because $g_i$ is only an estimate of the true gradient of the loss function at $w_i$. 
For deep neural networks, the gradients can be computed by the backpropagation algorithm~\cite{backprop} that requires a forward and a backward evaluation of the model.

\subsection{Synchronous SGD}
As models get larger, training requires parallelizing them on distributed hardware.
While there are other important sources of parallelism such as model parallelism and pipeline parallelism, these techniques are orthogonal to the data parallelism studied in this paper.

A common approach to data parallelism is {\em synchronous SGD}, where one computes the gradients for each dataset in a minibatch in parallel. This process works as follows. 
Each node processes a {\em microbatch} of data and computes its local gradient. The microbatch size is usually determined by the amount of memory available in the local node. 
Then a communication step sums up the local gradients from all nodes to compute the minibatch gradient. This process is called {\em allreduce}, after the MPI primitive that computes the sum of vectors from each node and stores the result in all the nodes.  Optionally, each node can perform {\em gradient accumulation} to sum up multiple microbatches before communicating with other nodes.
\adasum, the technique proposed in this paper, replaces allreduce for improved parallelism. 

\subsection{Algorithmic Efficiency vs System Efficiency}
One way to increase data parallelism is to increase the minibatch size. However, this has the effect of reducing stochasticity of SGD resulting in reduced convergence. Thus, there is a delicate balance between parallelism and model performance in distributed ML training. To capture this tradeoff, we define two notions. {\em System Efficiency} represents the raw throughput of the system in the amount of training data processes per unit of time. Obviously, naively increasing the minibatch size will increase the system efficiency of synchronous SGD. {\em Algorithmic Efficiency} is the inverse of the amount of training data that needs to be processed in order to achieve some desired model accuracy. Increasing the minibatch size decreases the algorithmic efficiency requiring more data or equivalently more iterations/epochs on a given training dataset to achieve the same level of accuracy. In the worst case, training might not converge when minibatch size is increased beyond a certain threshold. 

As one essentially cares about the training time to desired accuracy, the net accuracy of distributed training is a combination of system efficiency and algorithmic efficiency.  

\subsection{Learning Rate Optimizers}
While one can naively use the resulting gradient to update the model, researchers have proposed various learning-rate {\em optimizers} that adaptively use different learning rates for different parts of the model. For instance, Adam~\cite{adam} computes individual adaptive learning rates for different parameters based on the estimates of first and second moments of gradients. LARS~\cite{lars} uses layer-wise adaptive learning rates for greater training stability. The recent LAMB optimizer~\cite{lamb} extends LARS to effectively train models like BERT with large minibatch sizes. The benefits of \adasum{} are orthogonal to and can be used in concert with these optimizers.

\section{\adasum Algorithm}\label{sec:algo}
We present the intuitions behind the \adasum algorithm, while a mathematical treatment with convergence proof is in the Appendix. 

Consider two nodes that respectively compute gradients $g_1$ and $g_2$ respectively on minibatches $b_1$ and $b_2$. When using synchronous SGD, the effective gradient is the average $(g_1 + g_2)/2$. But as it is common to increase the learning rate proportional to the increased effective batch size, the combination amounts to a sum in practice. The main proposal behind this paper is to use an {\em adaptive sum} of the two gradients called the \adasum. 

$$ \adasum(g_1, g_2) = (1 - \frac{g_1^T \cdot g_2}{2\cdot \|g_1\|^2}) g_1 + (1 - \frac{g_1^T \cdot g_2}{2\cdot \|g_2\|^2})g_2 $$

Despite its apparent complexity, \adasum simply adds the two gradients after scaling with appropriate scalars. Using this operation instead of sum or average has the following two properties. First, the \adasum operation approximates the sequential execution of the two nodes running one after the other, thereby achieving convergence properties of smaller minibatch sizes. Additionally, it samples {\em both} possible orders of visiting minibatches --- $b_1, b_2$ and $b_2, b_1$. 
We then show how to extend the \adasum operation to more than two minibatches.  

\subsection{Emulating Sequential SGD}
Consider two steps of SGD starting from model $w_0$. Say, the first step computes a gradient $g_1(w_0)$ of the loss function at $w_0$ for $b_1$. 
SGD updates the model to 
$$w_1 = w_0 - \alpha \cdot g_1(w_0)$$ 
The second step computes its gradient $g_2(w_1)$ at $w_1$ for $b_2$. Assuming we are using the same learning rate for both steps, the final model after the second step is 
\begin{equation}
  \begin{split}
  w_{1,2} &= w_1 - \alpha \cdot g_2(w_0) \\
  & = w_0 - \alpha \cdot (g_1(w_0) + g_2(w_1))
  \end{split}
  \label{eq:seq}
\end{equation}
Here the subscript for $w$ indicates that SGD processed $b_1$ before $b_2$. 
Comparing this what a synchronous SGD algorithm would have computed (with a corresponding doubling of the learning rate) 
$$w_0 - \alpha \cdot (g_1(w_0) + g_2(w_0))$$ 
we see a difference because gradient $g_2$ is computed at $w_0$ instead of $w_1$. As previously observed~\cite{msra, ouripdps}, one can use second order reasoning to remove this staleness. Neglecting higher order terms in the Taylor expansion of $g_2(w_1)$, we have
\begin{equation}
  g_2(w_1) = g_2 - \alpha \cdot H_2 \cdot g_1
  \label{eq:seqhess}
\end{equation}
where $H_2$ is the Hessian matrix of the loss function at $w_0$. 

\subsection{Eliminating Hyperparameters}
One nice property of the standard synchronous SGD is that it adds no hyperparameters during gradient combination - we simply sum or average the two gradients. It is desirable to have the same property for \adasum.  
First, using a standard theorem~\cite{ggt, ggtblog} for estimating the Hessian matrix for negative log likelihood loss functions~(details in Appendix~\ref{app:taylor}), we can estimate
\begin{equation}
  g_2(w_1) = g_2 - \alpha \cdot g_2 \cdot g_2^T \cdot g_1
\end{equation}
where for simplicity we have dropped the model term from $g_1$ and $g_2$ when computed at $w_0$ for terseness.  
This along with the assumption that $\alpha$ is chosen optimally~(Appendix~\ref{app:optimallr}), we have
\begin{equation}
g_2(w_1) = g_2 - \frac{g_2^T \cdot g_1}{\norm{g_2}^2} g_2
\end{equation}
Essentially, by scaling $g_2$ with an appropriate scalar, we can emulate the sequential execution of two SGD steps in Equation~\ref{eq:seq} as
\begin{equation} \label{eq:asyncadasum}
  w_{1,2} = w_0 - \alpha \cdot [g_1 + (1 - \frac{g_2^T \cdot g_1}{\norm{g_2}^2}) \cdot g_2]
\end{equation}

\subsection{Sampling Multiple Paths}
The ability to emulate a sequential execution shown above provides an intriguing possibility. SGD is a stochastic process that samples a {\em path} defined 
by the order of the training data it processes. Now, the emulation above provides us a way to sample multiple paths for the cost of one! For instance, if SGD had processed
minibatch $b_2$ before $b_1$, the final model would be
 $$ w_{2,1} = w_0 - \alpha \cdot [g_2 + (1 - \frac{g_2^T \cdot g_1}{\norm{g_1}^2}) \cdot g_1]$$
Averaging the two samples, the final model would be
\begin{equation*}
  \begin{split}
w_{1,2+2,1} & = w_0 - \alpha \cdot [(1 - \frac{g_2^T \cdot g_1}{2 \cdot \norm{g_1}^2}) \cdot g_1 + (1 - \frac{g_2^T \cdot g_1}{2 \cdot \norm{g_2}^2}) \cdot g_2]\\
            & = w_0 - \alpha \cdot \adasum(g_1, g_2)\\
  \end{split}
\end{equation*}
This equation motivates our design of the \adasum operator. 

\subsection{Combining More Than Two Gradients}
\label{sec:binarytree}
We can extend \adasum to more than two gradients by recursively applying the operator as follows. Let $g_{[0,n]}$ be the result of applying \adasum to minibatches $b_0 \ldots b_{n}$. We can consider this as the {\em effective} gradient of these minibatches when emulating the behavior of SGD on $b_{n+1}$. Using the same arguments as above, we have 
$$ \adasum(g_{[0,n+1]}) = \adasum(\adasum(g_{[0,n]}), g_{n+1})$$
Since we double the number of paths of SGD emulated at each step, we achieve the effect of emulating {\em exponentially} many SGD paths. 

As discussed in Section~\ref{sec:impl}, one can reuse the standard ring algorithm used to sum all the gradients in synchronous SGD to implement the \adasum operation. 
One complexity is that the operation cannot be performed in a streaming manner as we need to compute the dot product and norm of the two gradients. For a more bandwidth optimal implementation, we use the following recursive application in practice.   

$$ \adasum(g_{[0,n]}) = \adasum(\adasum(g_{[0,n/2)}), \adasum(g_{[n/2,n]}))$$

\subsection{\adasum Properties}
Consider $\adasum(g_1, g_2)$ when $g_1$ and $g_2$ are orthogonal. Their dot product $g_2^T \cdot g_1$ is zero. Therefore, \adasum simply adds the two gradients. 
Now consider the case when they are parallel. Their dot product is simply the product of their norms. So, \adasum becomes the average of the two gradients. 
Intuitively, when the two gradients are pointing in orthogonal directions, \adasum behaves as if their loss functions are locally independent and aggressively sums the two gradients. Doing so when the two gradients are parallel has the danger of "overshooting" the minimum, particularly when the learning rate is also aggressive and therefore,
 \adasum safely averages the gradients.
This adaptiveness becomes important as we later show that gradients from different batches tend to point in the same direction during the initial parts of the training. This is because the initial model is completely random
and all gradients agree on the general direction model should progress. However, the gradients
 become progressively orthogonal in later parts of the training. \adasum{} automatically and adaptively interpolates between an aggressive sum and a safe average as training proceeds.

\subsection{Per-Layer Orthogonality}
For further insights, Figure~\ref{fig:shadow} shows the per-layer orthogonality of gradients during training. At different points during the training of \resnet and \bertlarge with $64$ GPUs, we compared the norm of the result of \adasum on the gradients from individual layers across all $64$ nodes with the individual norm of the gradients. Orthogonality of a set of gradients $g_1 \ldots g_n$ for a given layer is defined as $\frac{\norm{\adasum(g_{[1,n]})}^2}{\sum_i \norm{g_i}^2}$. This compares the norm of the result of \adasum on a set of gradients with the sum of their individual norms. Because of the properties of \adasum described above, orthogonality is $1$ when the gradients are orthogonal to each other and reaches the minimum value of $1/64$ when the gradients are parallel to each other and are of the same norm. Of other ways of measuring orthogonality, this was the easiest for us to collect experimentally. 

\begin{figure}[tb]
  \begin{subfigure}[b]{\columnwidth}
    \includegraphics[width=\columnwidth]{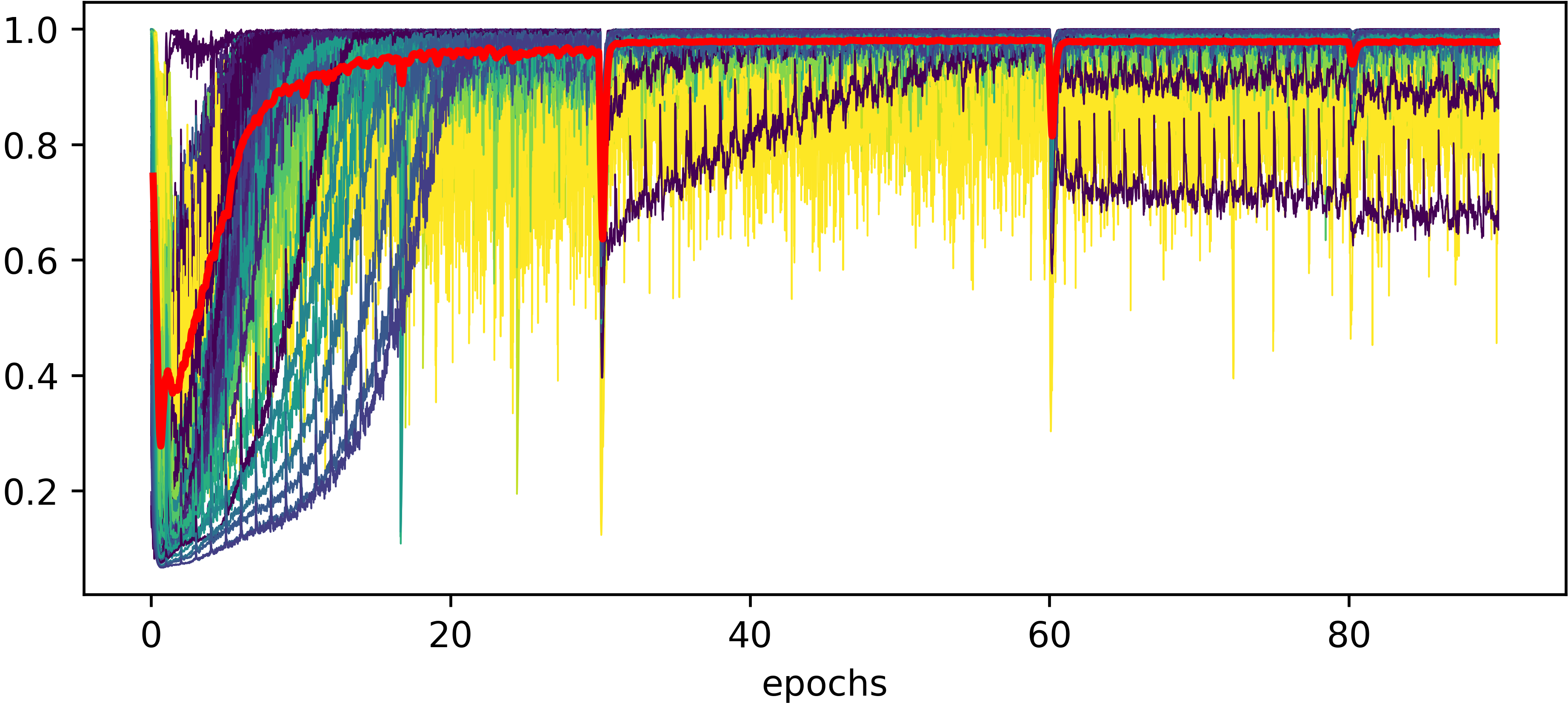}
    \caption{\resnet}
    \label{subfig:resnetlayer}
  \end{subfigure}

  \begin{subfigure}[b]{\columnwidth}
    \includegraphics[width=\columnwidth]{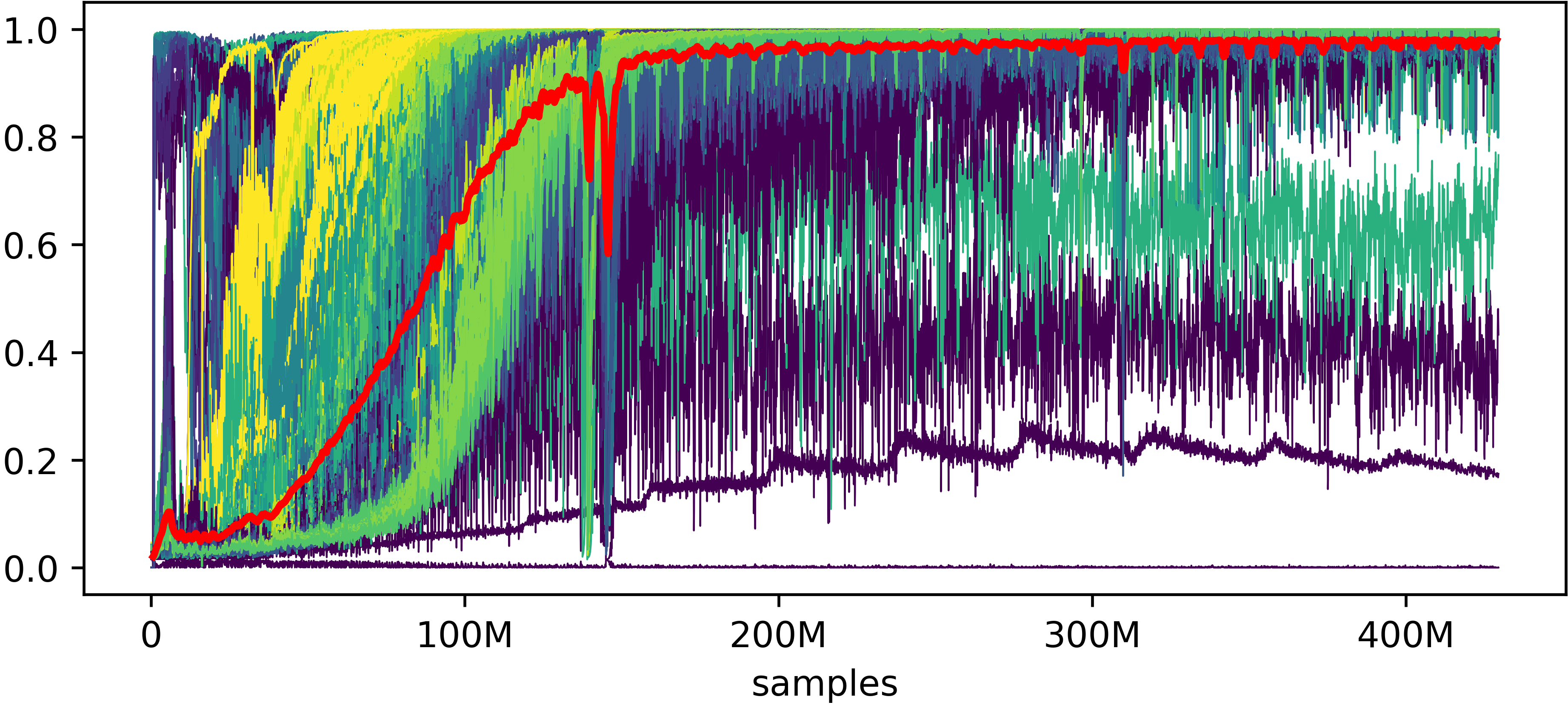}
    \caption{\bertlarge}
    \label{subfig:bertlayer}
  \end{subfigure}
    \caption{Orthogonality of gradients during \resnet~(a) and \bertlarge~(b). x-axis represents number of samples or epochs processed during training. A value of $1$ in the y-axis means that the gradients are orthogonal with lower values meaning less orthogonality. The bold red lines show orthogonality averaged across all layers, while the others show the orthogonality of individual layers with different colors.}
  \label{fig:shadow}
\end{figure}

The figures~\ref{subfig:resnetlayer} and \ref{subfig:bertlayer} show the orthogonality of different
layers and their average, shown by the bold red lines, for both \resnet and \bertlarge. We can clearly see from the
average of the orthogonality (red lines) that the gradients start out pointing in the same direction but very soon become 
orthogonal as the training proceeds. Each layer also demonstrate a similar pattern with a different color shown
in Figure~\ref{fig:shadow}. Of course, there are too many layers to distinguish each color individually. However, general trends are still visible. While most layers tend to become orthogonal as training proceeds, they do not do so at the same rate. This discrepancy is more visible for \bertlarge, where some layers have low orthogonality throughout the training process. Note that, there are clear drops in the orthogonality during the training for both benchmarks. These drops
happens exactly at boundaries of learning rate schedule change. 

To exploit this observation, we perform the \adasum operation {\em per layer} as opposed to applying to the whole gradient. This allows us to adaptively adjust the combination based on per-layer orthogonality.




\subsection{Evaluating Sequential Emulation}
\begin{figure}
\centering
  \includegraphics[width=\columnwidth]{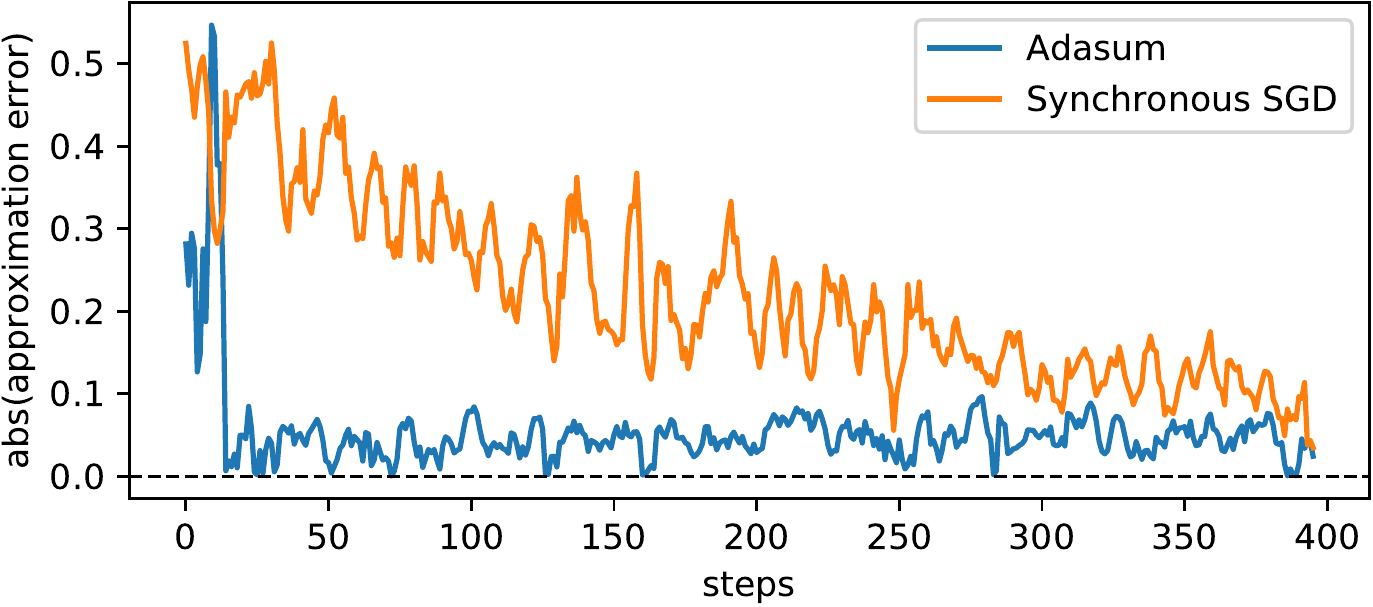}
  \caption{Relative error of \adasum and synchronous SGD when compared to a sequential emulation that uses the exact Hessian matrix. }
  \label{fig:error}
\end{figure}
To validate our intuitions behind \adasum, we empirically evaluate how closely \adasum emulates a sequential emulation.
As shown earlier in Equation~\ref{eq:seqhess}, we can reduce the staleness of gradients using the Hessian of the loss function. Luckily, for small models such as LeNet-5, 
we can compute the Hessian matrix exactly. 
Specifically, we downloaded the \torch MNIST tutorial example, modified the code to ensure deterministic runs, and used 
\torch's autograd facilities to compute the exact Hessian at each step during a parallel run with 64 nodes (that reached the target accuracy of 99.3\%). 
After every communication step, we computed the model with sequential emulation using the exact Hessian, using the \adasum operator, and using the baseline synchronous SGD. Figure~\ref{fig:error} shows the relative error of \adasum and synchronous SGD with respect to the sequential emulation using the Hessian. As we can see, \adasum has a  lower approximation error reaching close to zero at some steps. Note synchronous SGD error goes down with number of steps---
the reason is because the norm of gradients $\norm{g}$ decays as the model approach the optimal answer
and $H\approx g\cdot g^T$ decays quadratically. In a real run, the error shown here is accumulated during the training process. This further validates our intuition that \adasum should achieve faster convergence than synchronous SGD, which we demonstrate empirically in Section~\ref{sec:results}.

\section{Implementation}\label{sec:impl}
\adasum is implemented in Horovod~\cite{horovod} and is publicly available in the main branch of its open-source repository. 
Horovod integrates with multiple machine learning frameworks, such as \tf and \torch. 
Horovod has a \verb!C++! backend targeting multiple backend transports like Ethernet/IB or
NVLINK. \adasum works with CUDA aware MPI (when available) and we implemented \adasum for all of these backends. 

\subsection{Using \adasum in Horovod}
\adasum is easy to use by specifying an option to the {\tt DistributedOptimizer} API of Horovod as follows. 
\begin{lstlisting}[language=Python, aboveskip=0.5\baselineskip, belowskip=0\baselineskip]
opt = hvd.DistributedOptimizer(opt, op=hvd.Adasum)

\end{lstlisting}
When enabled, the distributed optimizer calls the necessary \adasum allreduce operations to
synchronize global model updates. As with Horovod, the user is responsible for partitioning data across nodes and initializing the model correctly in all nodes. 

Providing the option is the only change required for users wanting to use an existing {\tt DistributedOptimizer} such as Adam or LAMB.  
For more fine grained control, we also expose the \adasum operator through Horovod's {\tt allreduce} using the same option. 
\begin{lstlisting}[language=Python, aboveskip=0.5\baselineskip, belowskip=0\baselineskip]
  hvd.allreduce(opt, op=hvd.Adasum)
\end{lstlisting}
This is useful when users want to perform additional operations such as gradient clipping beyond those implemented in a {\tt DistributedOptimizer}. 

\begin{figure}
  \begin{lstlisting}[language=Python]
    params = list(model.parameters())

    # save a copy of the model parameters
    starts = [p.clone().detach() for p in params]

    # do forward / backward
    ...
    
    # optimizer step
    optimizer.step() # modifies params  

    # apply adasum after optimizer
    for start, current in zip(starts, params):
      effective_gradient = current - start
      hvd.allreduce(effective_gradient, op=hvd.Adasum)
      current.data.add_(effective_gradient)
  \end{lstlisting}
  \caption{Implementation of \adasum with optimizers such as Adam or LAMB.}
  \label{fig:adasumopt}
  \end{figure}


One subtlety in the implementation is that the \adasum operation should be performed {\em after} the optimizer update, as shown in Figure~\ref{fig:adasumopt}. 
In contrast, synchronous SGD performs allreduce before the optimizer update. Intuitively, this is because \adasum does not increase the minibatch size like synchronous SGD when distributing across multiple nodes. As such, the logic of optimizers should only apply to the smaller minibatches per node. 
This is similar to BMUF~\cite{chen2016scalable}. Users have to replicate this logic explicitly when not using an existing {\tt DistributedOptimizer} directly.

\subsection{\adasum{} Allreduce Implementation}
This section describes the implementation of the \adasum operator in Horovod's allreduce. 
Message Passing Interface (MPI) provides capabilities to perform user-defined reduction operations for allreduce. 
However, since these custom reductions can only be elementwise, \adasum cannot be implemented as a user-defined reduction.


\newcommand{\Allreduce}{\textsc{Allreduce}}
\newcommand{\RVHAdasum}{\textsc{AdasumRVH}\xspace}
\newcommand{\Send}{\textsc{Send}}
\newcommand{\Recv}{\textsc{Recv}}
\newcommand{\GroupComm}{\textsc{GroupComm}}
\newcommand{\assign}{=}
\newcommand{\rank}{\mathit{rank}}
\newcommand{\ranks}{\mathbb{N}^{<\size}}
\newcommand{\size}{\mathit{size}}
\newcommand{\dist}{\mathit{d}}
\newcommand{\nghr}{\mathit{nghr}}
\newcommand{\midvar}{\mathit{mid}}
\newcommand{\group}{\mathit{group}}

\subsubsection{Recursive Vector-Halving}
\begin{algorithm}[tb]
\caption{Recursive vector-halving with \adasum\label{alg:vhdd-adasum}}
\begin{algorithmic}[1]
    \Require $\size>2$ is a power-of-two.
    \Procedure{$\RVHAdasum$}{$x$,$\dist$}
        \State $\midvar\assign\lfloor|x|/2\rfloor$\label{alg:vhdd-adasum:line:start-exhange}
        \If{$\lfloor\rank/\dist\rfloor$ is even} \Comment{Left neighbor}
            \State $\nghr\assign\rank+\dist$
            \State $\Send(x_{\midvar:|x|},\nghr)$ \Comment{Send right half}
            \State $a\assign x_{0:\midvar}$
            \State $b\assign \Recv(\nghr)$ \Comment{Receive left half}
        \Else \Comment{Right neighbor}
            \State $\nghr\assign\rank-\dist$
            \State $\Send(x_{0:\midvar},\nghr)$ \Comment{Send left half}
            \State $a\assign \Recv(\nghr)$ \Comment{Receive right half}
            \State $b\assign x_{\midvar:|x|}$
        \EndIf\label{alg:vhdd-adasum:line:end-exhange}
        \State $\dist'\assign 2\cdot\dist$
        \State $v\assign[a\cdot b,a\cdot a,b\cdot b]$\label{alg:vhdd-adasum:line:partial} \Comment{Partial dot products}
        \State $\group\assign[\lfloor\frac{\rank}{\dist'}\rfloor\cdot \dist'+i\text{ for }i\assign 0\twodots \dist'-1]$\label{alg:vhdd-adasum:line:group}
        \State $v\assign\Allreduce{}(v,+,\group)$\label{alg:vhdd-adasum:line:allreduce} \Comment{Finish dot products}
        \State $x'\assign a\cdot(1-\frac{v_1}{2v_2}) + b\cdot(1-\frac{v_1}{2v_3})$\label{alg:vhdd-adasum:line:finish} \Comment{Apply \adasum}
        \If{$\dist'<\size$}
            \State $x'\assign \RVHAdasum(x',\dist')$\label{alg:vhdd-adasum:line:recurse}
        \EndIf
        \State $\Send(x',\nghr)$ \Comment{Send my half}\label{alg:vhdd-adasum:line:start-allgather}
        \State $y\assign\Recv(\nghr)$ \Comment{Receive neighbor's half}
        \State $x\assign x'\mdoubleplus y$ \textbf{if} $\lfloor\rank/\dist\rfloor$ is even \textbf{else} $y\mdoubleplus x'$\label{alg:vhdd-adasum:line:end-allgather}
    \EndProcedure
\end{algorithmic}
\end{algorithm}

Our implementation of \adasum uses a modified recursive vector-halving~(RVH) algorithm for
allreduce~\cite{Vandegeijn94, Chan07}, which is both latency and bandwidth
optimal in a hypercube and fully connected networks. On each step of the
reduce-scatter phase of the algorithm each node exchanges half of its data
with its neighbor and applies the reduction on its own half. 
In the baseline algorithm, since each application of the reduction operation is only given a part of the data, the operation must be elementwise to ensure the correct result.

To work around this problem we modify the RVH algorithm to
perform each \adasum{} operation in two phases, with an additional step of
communication in between. Algorithm~\ref{alg:vhdd-adasum} describes these
modifications. Each process is identified by a zero-based index called
rank. A process can send a vector $v$ to another rank $r$ with $\Send(v,r)$ and
receive one with $v\assign\Recv(r)$. The algorithm also uses another allreduce
as a primitive to sum partial dot products and squared norms across subgroups of
ranks: $\Allreduce(v^i,\mathit{op},\group)$ returns a pointwise reduction of
vectors $v_i$ from all ranks $i\in\group$, where $\group$ is a list of the ranks
participating in the reduction

Each level of recursion in Algorithm~\ref{alg:vhdd-adasum} starts with ranks
exchanging half of their vector $x$ with a neighbor at distance $\dist$
(lines~\ref{alg:vhdd-adasum:line:start-exhange}-\ref{alg:vhdd-adasum:line:end-exhange}).
Here the two halves $a$ and $b$ are assigned such the left neighbor's half is in
$a$ and the right neighbor's half is in $b$.

Lines~\ref{alg:vhdd-adasum:line:partial}-\ref{alg:vhdd-adasum:line:finish} of
Algorithm~\ref{alg:vhdd-adasum} represent the main modification to baseline RVH algorithm. 
First, on line~\ref{alg:vhdd-adasum:line:partial}, each rank
calculates a dot product and squared norms for $a$ and $b$, which are slices of
a larger logical vector shared across exactly the ranks in $\group$
(line~\ref{alg:vhdd-adasum:line:group}).
Line~\ref{alg:vhdd-adasum:line:allreduce} then sums the products among the ranks
in $\group$ to produce the complete results in $v$. The reduction is finally
applied locally using the values in $v$
(line~\ref{alg:vhdd-adasum:line:finish}).

The algorithm continues as normal, using recursion on
line~\ref{alg:vhdd-adasum:line:recurse} until all ranks share slices of the same
reduced vector, followed by an all-gather phase on
lines~\ref{alg:vhdd-adasum:line:start-allgather}-\ref{alg:vhdd-adasum:line:end-allgather}.

\subsubsection{Integration into Horovod} \label{sec:int:horovod}
Horovod uses \RVHAdasum{} to reduce tensors whenever \lstinline{hvd.allreduce}
or \lstinline{hvd.DistributedOptimizer} is used with \lstinline{op=hvd.Adasum}.
Additionally, if the {\tt HOROVOD\_HIERARCHICAL\_ALLREDUCE environment} variable is
set Horovod performs a \emph{hierarchical} allreduce using the NVIDIA Collective
Communications Library (NCCL). This variant starts and ends with a NCCL
reduce-scatter and allgather, respectively, for communication among the GPUs
inside a node, with cross-node reduction handled by \RVHAdasum{}. This is useful
with some hardware configurations on which NCCL offers higher throughput than
CUDA-aware MPI. 




\subsubsection{\RVHAdasum Performance}
\label{sec:performance}

Now we evaluate the latency
of $\RVHAdasum$ on 16 Azure nodes with 4 V100s per node (PCIe
interconnect) and a single 100 Gb/s Infiniband connection between
them.  As a baseline, we compare to NCCL's sum operation.
\begin{figure}
  \includegraphics[width=\columnwidth]{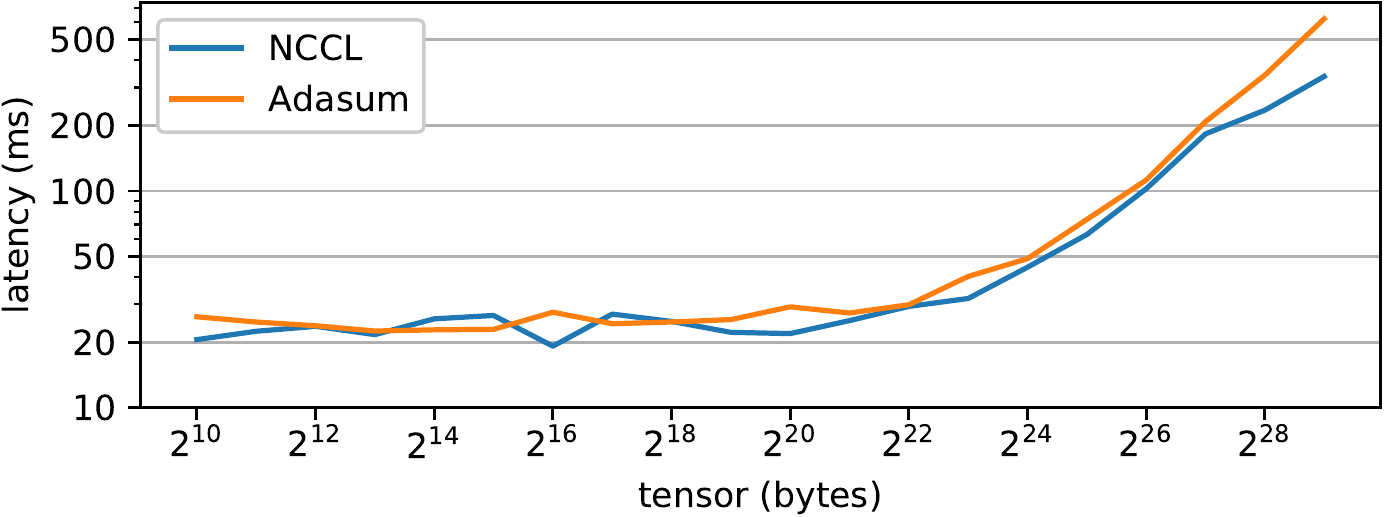}
  \caption{Latency of $\RVHAdasum$ vs. NCCL for various message sizes.\label{fig:bandwidth}}
\end{figure}
A point on Figure~\ref{fig:bandwidth} ($x$,$y$) shows the latency for
an allreduce operation ($y$) in seconds as a function of the number of
bytes reduced ($x$).  For each point on the $x$ axis, we allocate 64
tensors on each GPU's memory so their sum is the the number of bytes.  The figure
demonstrates that despite the additional logic required to perform an
adaptive summation, the performance of $RVHAdasum$ is roughly equal to
the highly optimized NCCL library simply doing a summation. Note
that these results use vectorization as well as tensor fusion
with a threshold of 2MB as discussed in Section~\ref{sec:vect} and \ref{sec:fusion}.

As section~\ref{sec:binarytree} described, there are two ways of applying the pairwise \adasum operation on a set of gradients.
\RVHAdasum performs the "tree" reduction. An alternate way is to apply the \adasum operator linearly. 
We additionally implemented this approach and optimized it using techniques similar to the ring allreduce algorithm commonly used for synchronous SGD. 
We found that this "ring" implementation provided less throughput than the baseline NCCL allreduce and \RVHAdasum on the architectures we evaluated. 
Nevertheless, we believe the ring allreduce version of \adasum could be competitive for other architectures.

\subsection{Parallelizing \adasum Computation}
\label{sec:part}

For large models such as \bertlarge, memory available in a GPU only fits a small microbatch size. 
In such cases, to increase the effective microbatch size, we use the GPUs available in a single node to accumulate local gradients and use the \adasum operation across nodes. In these scenarios, we parallelize the \adasum computation across all these local GPUs.  

Our approach is inspired by the optimizer-state partitioning algorithm pioneered by Marian~\cite{marian}, a deep learning toolkit optimized for NLP
workloads. Optimizers like Adam or LAMB maintain additional state per gradient element to estimate moving mean and/or variance. 
The optimizer parameters are identical for all GPUs and thus it is not necessary to replicate them. Marian partitions this state and parallelizes the updates to the state. 
Note that this memory optimization is not related to model parallelism as model parameters are still replicated on all GPUs.

We use the same insight to parallelize the \adasum computation. Looking at Figure~\ref{fig:adasumopt}, we partition the optimizer state as in Marian. In addition,
we also partition the {\tt effective\_gradient} across the local GPUs. A key difference between the Marian approach is that rather than distributing this state uniformly, 
we partition to ensure that state corresponding to one neural network layer falls in the same partition. This greatly simplifies our implementation as we do not have 
to modify the code of the underlying optimizer. 

After the optimizer update, which is done in parallel (as in Marian), the {\tt effective\_gradient} in Figure~\ref{fig:adasumopt} is already partitioned across local GPUs. 
Now, each GPU does an \adasum allreduce only on the layers in its partition by communicating with the corresponding GPU that share partitions in other nodes. 
Finally, the GPU broadcasts its partition of {\tt effective\_gradient} locally to all other GPUs in the same node to update the model parameters.   
To optimize the cost of this local broadcast, we overlap this communication with the \adasum operation of the next layer. 

\begin{table}[h]
  \small
  \begin{center}
  \begin{tabular}{@{}llll@{}}

  & Without & With \\
  \midrule
  Throughput (samples/s) & 154.7 & 168.5 \\
  Model update (s) & 1.82 & 0.97 \\
  Microbatch & 22 & 36 \\
  \end{tabular}
\end{center}
  \caption{Performance improvement with and without \adasum parallelization.}
  \label{table:part}
  \end{table}

The table above shows the performance improvement of this optimization for a \torch implementation of \bertlarge on a Azure VM with 4
  V100s (16 GB RAM) connected by PCIe for 128 max sequence length. 
Since this optimization reduces the memory usage, we can increase the microbatch size by 60\% as shown in the last column. 
 To measure the impact of this larger microbatch,
  we evaluate 256 microbatches before an adasum operation 
  The first column in
  Table~\ref{table:part} provides the throughput per GPU. In other
  words, the 60\% larger microbatch 
  yields nearly a 10\% improvement in per GPU throughput.
  To measure the impact of parallelizing the \adasum operation, we evaluate a
  microbatch size of 1 per model update.  The third column in
  Table~\ref{table:part} shows that this time drops 
  nearly 1.87$\times$.


\subsection{Implementation Details}
This section describes implementation details that are crucial for improving system and algorithmic efficiency of \adasum. 

\subsubsection{Low Precision Support}
Recent trends in ML training have shown the promise of using low-precision formats such as 
\texttt{fp16} for compute and communication efficiency. 
Our implementation of \adasum integrates with the low-precision support in Horovod to obtain these benefits automatically. 

We discuss two important subtleties in our implementation of low-precision support. 
First, \adasum requires computing the dot product and norms of the combined gradients. 
The accumulation of the values for these two operations happens in a {\tt double} even if the gradients are in lower precision. 
This does not incur any measurable overhead for both CPU and GPU implementations. On the other hand, the improved floating point stability is crucial for the improved convergence of \adasum.  

When using lower-precision, it is common to use {\em dynamic scaling}~\cite{mixedprecision}.
The basic idea is to maintain a {\em scale} for all tensors to ensure that the values are always in the dynamic range of the low-precision format. 
During the training, these scales have to be periodically adjusted when the values exceed the range (resulting in {\tt nan}s). We perform dynamic scaling for tensors we introduce, such as the {\tt effective\_gradient} in Figure~\ref{fig:adasumopt}.

\subsubsection{CPU and GPU Vectorization}
\label{sec:vect}
\adasum{} runs on both CPU and GPU hardware in \texttt{fp16},
\texttt{fp32}, and \texttt{fp64}.  For CPU hardware, we manually
vectorize loop bodies that perform both dot products and summations.
When Horovod is compiled with CUDA aware MPI, we implement these
same loops as GPU kernel calls that operate directly on GPU
memory and thus save on the transfer from GPU to CPU.  This is
particularly important on hardware that supports GPUDirect RDMA as GPU
memory need not be copied to the CPU for the \adasum{} operator.

\subsubsection{Tensor Fusion}
\label{sec:fusion}
If at the time of allreducing a tensor, tensors from other layers are
also ready and available on all hosts, Horovod fuses these tensors
into a single one, performs an allreduce on the fused tensor, and then
copies from the fused tensor back to the individual tensors.  This
optimization significantly reduces latency for small tensors as the
overhead of potentially many individual allreduce calls are amortized
into a single one.  To enable this optimization with \adasum, we do
additional bookkeeping to keep track of tensor boundaries in the fused
tensor as \adasum{} requires these boundaries to compute dot products
per layer.  Because all hosts 1) fuse the same set of tensors and 2)
have the same layer sizes, this bookkeeping is stored locally and does
not increase communication overheads. Note that Horovod uses an
extra buffer for copying the tensors from different layers
into a consecutive array so that underlying libraries such as MPI
or NCCL can be called once. \adasum uses the same buffer for this
optimization. The size of this buffer is controlled by {\tt
HOROVOD\_FUSION\_THRESHOLD}. A default value between 2MB-64MB
usually works well.
\section{Results}\label{sec:results}
To show that \adasum works across a variety of real-world training scenarios, 
we evaluate its performance through a sequence of case studies.
Throughout the experiments, we use our implementation of \adasum in Horovod described in Section~\ref{sec:impl} on models implemented in both \torch and \tf. 

We first study the algorithmic and system efficiency of \adasum for \resnet and \bertlarge. 
Then, we use LeNET5, which is small enough to do extensive hyperparameter tuning, to
show that \adasum{} enables robust scaling without the need for additional hyperparameter tuning. 
Finally, we finish with a short summary of our experience with production models. 



\subsection{\resnet on Fast Interconnects}
This section evaluates \adasum on 
\torch's \resnet using
the Momentum-SGD optimizer
on hardware with a fast interconnect.

\subsubsection{Platform and Methods}
\resnet~\cite{resnet50} on Imagenet\cite{ILSVRC15} is a popular model
for studying performance of training algorithms and
implementations. We used \torch's \resnet model modified to run with
Horovod and compared the performance of \adasum with Horovod's default
$Sum$ operator as the baseline for synchronous SGD. We ran experiments
on Azure's \texttt{Standard\_NC24rs\_v3} virtual machines, each of
which has 4 NVIDIA Tesla V100 GPUs connected with PCIe, dual-socket
Intel Xeon E5-2690 v4 CPUs, 448 GiB of memory, and connected via
Infiniband. Hierarchical allreduce, as described in
Section~\ref{sec:int:horovod} was faster for both the baseline and
\adasum{} runs so we used it here.



We train on 64 V100s with 2K and 16K examples per allreduce and use
the default hyper-parameters that ship with the benchmark for its
momentum based SGD optimizer.
%

\subsubsection{Algorithmic Efficiency}

The number of epochs required for each configuration to reach the target
accuracy are as follows:
\begin{center}
  \small
\begin{tabular}{@{}llll@{}}
  \Horovodsum 2k & \Horovodsum 16k & \adasum{} 2k & \adasum{} 16k \\
  \midrule
  62 & - & 62 & 69 \\
\end{tabular}
\end{center}
Because \horovodsum with 16k batch size \emph{never} reaches 74.9\% validation accuracy
(we let it run for 120 epochs), it's algorithmic efficiency is zero. \adasum{}
on the other hand sees only a 11\% decline in its algorithmic efficiency as we increased the batch size. 
This is more than made up for in increased system efficiency.

\subsubsection{System Efficiency}
The times per epoch for each configuration are as follows:
\begin{center}
  \small
\begin{tabular}{@{}llll@{}}
  \Horovodsum 2k & \Horovodsum 16k & \adasum{} 2k & \adasum{} 16k \\
  \midrule
  5.61 min & 2.12 min & 5.72 min & 2.23 min \\
\end{tabular}
\end{center}
\adasum{} closely matches the system efficiency of Horovod's \horovodsum
implementation at both 2k and 16k batch size. Increasing the batch size from 2k
to 16k results in a 61\% and 62\% improvement for \adasum{} and \horovodsum,
respectively.

\begin{figure}
  \includegraphics[width=\columnwidth]{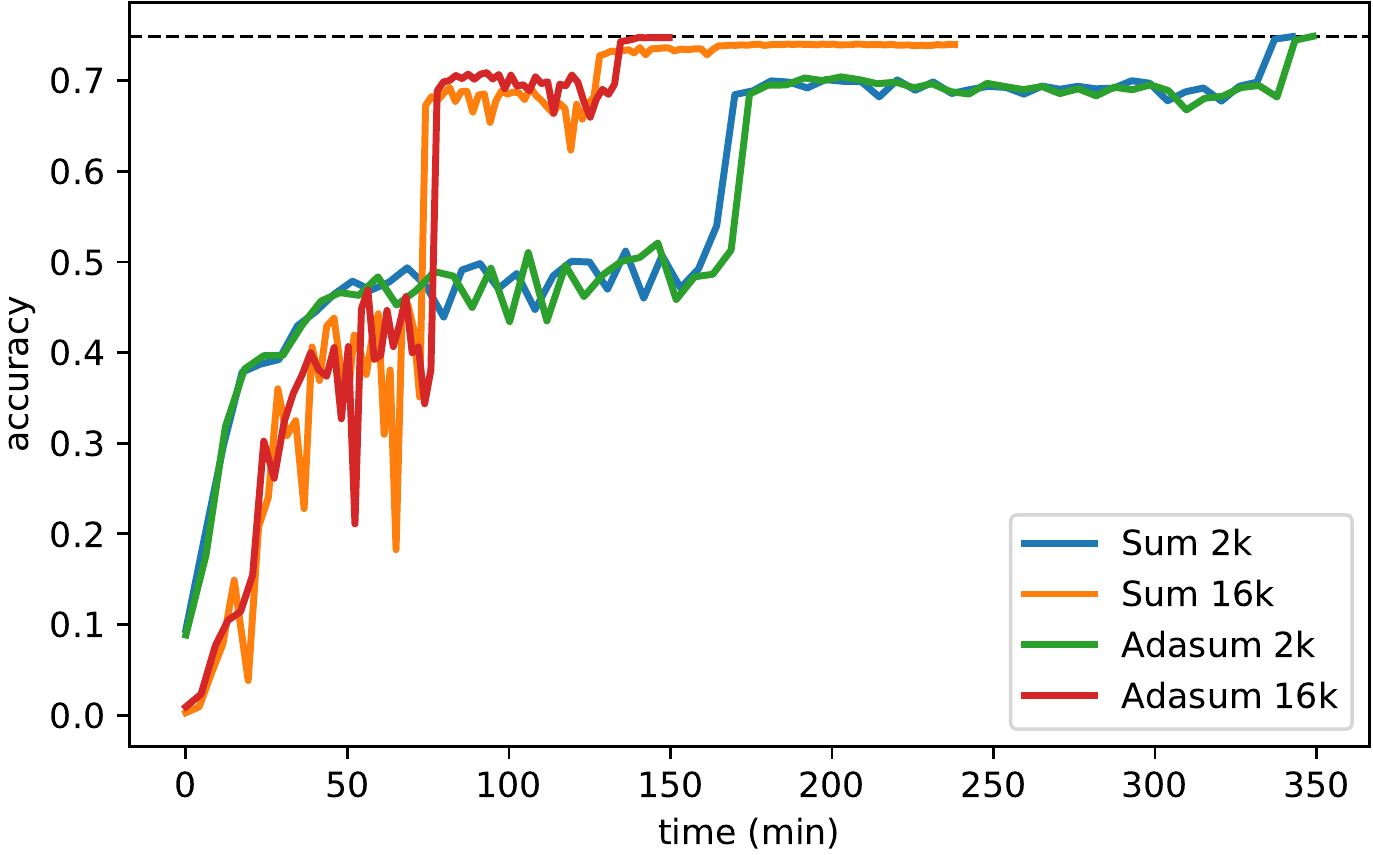}
  \caption{Time-to-accuracy chart for \resnet with 64 GPUs on 16 Standard\_NC24rs\_v3 VMs.}
  \label{fig:resnet50:tta}
\end{figure}
The total efficiency can be seen in Figure~\ref{fig:resnet50:tta}, which shows
the time ($x$-axis) to accuracy ($y$-axis) for each configuration. \Horovodsum
16k plateaus below the target accuracy. \adasum 16k, on the other
hand, gets to a top-1 accuracy of 74.9\%, being 2.3X
faster in time to accuracy than \adasum 2k, \emph{while using the same number of
GPUs}.

\subsection{\tf \resnet on Slow TCP}
Horovod and \adasum run on all types of hardware.  Often, networks
have TCP rather than IB and as such, communication limits scaling to
multiple nodes. This section demonstrates \adasum{} enables a larger
effective batch size, which reduces communication overhead, and yet
still converges fast enough to reduce time to accuracy.

\subsubsection{Platform and Methods}
This section demonstrates how a \tf implementation of \resnet from
MLPerf's v0.5 reference implementation~\cite{mlperf} is able to scale to 16 V100s
(4 GPUs 32GB cards per node with a PCIe gen 3 interconnect) and TCP
(40 GB/s) interconnects in between.

We downloaded the MLPerf v0.5 reference implementation and slightly
modified it to use Horovod with \adasum.  We did a small
hyper-parameter search over the learning rate but did not change any
other hyper-parameters.  We found $4\times$ the default learning rate
provided the fastest convergence.  This benchmark uses the Momentum
SGD optimizer from \tf. MLPerf v0.5 converges when the test accuracy
is greater than or equal to 74.9\%.  

Unlike simply adding the gradients together for gradient accumulation,
the \tf \adasum enabled distributed optimizer uses a local SGD step to
update weights; when it is time for an allreduce, the gradient is
estimated via a delta from the model's state since the prior
allreduce.
\begin{table}[t]
  \small
  \center{
    \begin{tabular}{@{}lrr@{}}
      Local steps before communicating & 16    & 1     \\ 
      \midrule
      Effective batch size             & 64K   & 4K    \\ 
      Minutes per epoch                    & 1.98  & 2.58  \\ 
      Epochs till convergence          & 84    & 68    \\
      \hline
      Time to accuracy(min)            & 166.32& 175.44\\
      \end{tabular}
    }
     \caption{Algorithmic and System efficiency for \tf \resnet on TCP.
      \adasum{} enables faster time to accuracy by doing more compute
      per communication.}
    \label{table:tfresnet}
\end{table}

\subsubsection{Algorithmic Efficiency}
Table~\ref{table:tfresnet} shows the results of this
experiment. \adasum for \tf enables a form of gradient accumulation
where the toolkit makes many local steps before communicating.  The
first row denotes how many local steps to make before initiating an
allreduce with the \adasum operation.  Note that when it takes 1 local
step before communicating, there is no gradient accumulation. In
contrast, 16 local steps before communicating means that the allreduce
is called once every 16 local steps.  With 16 GPUs and 256 examples
per GPU, the convergence is fast in 68 epochs (fourth row in
Table~\ref{table:tfresnet}).  That convergence slows to 84 epochs with
an effective batch size of 64K.  It is important to note that none of
the submissions to MLPerf v0.5 used more than a 16K effective batch
size (we verified by looking at the result submissions).  Thus
\adasum is able to scale the \tf momentum
optimizer to 64K examples before communicating $4\times$ more than prior
art. 

\subsubsection{System Efficiency}
When communicating after every step, the system efficiency is low as
communication dominates (third row in Table~\ref{table:tfresnet}).  In
contrast, when we communicate every 16 local steps, the time for 1
epoch is much faster. Total running time is given by the last row:
min per epoch * epochs till convergence and clearly communicating
less frequently has a big impact in overall running time despite the
slight increase in algorithmic efficiency. Thus, a developer can
exploit fast distributed hardware even without a fast interconnect.

\subsection{\bertlarge}
This section shows \adasum{} scales both Adam~\cite{kingma2014adam}
and LAMB~\cite{lamb} for the \torch NVIDIA implementation of
\bertlarge~\cite{nvidiabert}.

\subsubsection{Platform and Methods}
Training \bertlarge~\cite{bertorg} a natural language processing (NLP)
model, takes place in two stages.  First is an unsupervised
\emph{pre-training} on two large text corpuses, Wikipedia and
BookCorpus. Then, the model is fine-tuned for a ``downstream'' NLP
task such as \squad question-answering~\cite{squad1,squad2}. A target
F1 score of \squad 1.1 of 90.5 averaged over 5 tries with different
seeds is generally accepted for \bertlarge pre-training~\cite{lamb,bertorg}.

Pre-Training \bertlarge requires tokenizing input sentences into a max
sequence length. A maximum sequence length of 512 is computationally
expensive and thus Devlin \etal suggest \cite{bertorg} breaking
pre-training into two phases: phase 1 with a maximum sequence length of
128 for 90\% of the training iterations and phase 2 with a maximum
sequence length of 512 for the remaining 10\%.

NVIDIA's repo contains scripts that 1) download and pre-process data,
2) phase 1 and 2 pre-training and, 3) \squad fine-tuning and
evaluation. NVIDIA's pre-training scripts use mixed precision training
in addition to data parallelism with NCCL~\cite{nccl}.  This codebase
is our baseline and for the \adasum{} implementation, we replaced its
use of \lstinline[columns=fixed,language=Python]{torch.distributed}
with the \adasum operator in Horovod.

The system that we used for this case study is a cluster of DGX-2 nodes where each node has 16 V100 GPUs with 
32GB of memory per GPU and NVSwitch intra connection. Each node has 8 NICs with Infiniband support capable of 
delivering a throughput of 800GB/s per node.

\subsubsection{Algorithmic Efficiency}
\begin{table}[t]
  \small
  \center{
  \begin{tabular}{@{}lrr@{}}
    & \multicolumn{2}{@{}l@{}}{Number of iterations}\\
    Algorithm & Phase 1 & Phase 2  \\
    \midrule
    Baseline-Adam & - & - \\
    Baseline-LAMB~\cite{lamb} & 7039 & 1563 \\
    \adasum-Adam & 7039 & 1563 \\
    \adasum-LAMB - 20\% & 5639 & 1250 \\
    \adasum-LAMB - 30\% & 5039 & 1563 \\
    \adasum-LAMB - 128K & 4574 & 1563 \\
  \end{tabular}
  }
  \caption{Algorithmic efficiency results on \bertlarge. Table shows the number of iterations required for Phase 1 and Phase 2 to achieve target \squad score of 90.5, when using the effective batch size of 64K for Phase 1 and 32K for Phase 2.}
  \label{bertalgo}
\end{table}
Table~\ref{bertalgo} describes the algorithmic efficiency of \adasum over the Adam and LAMB optimizer when using an effective batch size of 64K for Phase 1 and 32K for Phase 2. As reported in prior work, the Adam optimizer does not scale to batch sizes beyond 16K. This motivated the study of more sophisticated optimizers such as LARS and LAMB. For instance, our runs of the LAMB optimizer achieve the target \squad score with $7039$ iterations of phase 1 and $1563$ iterations of phase 2, as shown in second row of Table~\ref{bertalgo}. 

The next two rows of Table~\ref{bertalgo} show the performance of \adasum. In contrast to Adam baseline, 
the \adasum-Adam optimizer converges with 64K when run with the same number iterations for Phase 1 and Phase 2 as the LAMB baseline. 
This is an interesting result as despite the advances of optimizers such as LAMB, Adam optimizer continues to be popular for some models. When compared to prior work~\cite{lamb}, it is important to note that \adasum adds no additional hyperparameters simply reusing the baseline 
parameters of the Adam optimizer. On the other hand, improvements provided by \adasum are orthogonal to improvements in optimizers. As shown in Table~\ref{bertalgo}, \adasum-LAMB provides close to 20\% faster convergence compared to LAMB baseline 
requiring 5639 iterations for phase 1 and 1250 iterations phase 2. 

We also performed two variations of our \adasum-LAMB results. First, we aggressively reduce the number of Phase 1 iterations by 30\%. With an equivalent aggressive reduction on Phase 2, we slightly missed the target \squad score by 0.5. However, we did achieve the target accuracy with the full 1563 iterations in Phase 2, which is what we report in the table. With a more fine grained search for Phase 2 iterations, we believe we can achieve the target \squad score with fewer Phase 2 iterations.
These results are still interesting as Phase 1 takes a larger percentage of training time than Phase 2. 

For the second variation, we increased the effective batch size of Phase 1 to 128K. We were able to achieve the target \squad score with 4574 Phase 1 iterations, while using the standard 1563 iterations of Phase 2 with 32K batch size. To the best of our knowledge, this is the largest report effective batch size for \bertlarge.   

\subsubsection{System Efficiency}
\begin{table}[t]
  \small
  \center{
  \begin{tabular}{@{}lllllll@{}}
    & \multicolumn{2}{l}{PH1 speedup} & \multicolumn{2}{l}{PH2 speedup} & \multicolumn{2}{l}{Time (minutes)} \\
    GPUs & Sum & Adasum &  Sum & Adasum &  Sum & Adasum \\ 
    \midrule
    64  & 1     &  0.98 & 1      & 0.99 & 997 & 809 \\
    256 & 3.79 &  3.61 & 3.89  & 3.92 & 260  & 214 \\
    512 & 7.47 &  6.48 & 7.24  & 7.28 & 135  & 118
  \end{tabular}
  }
  \caption{System efficiency on \bertlarge for an effective batch size of 64K and 32K for phase 1 and phase 2, 
  respectively. The speedup numbers are relative to the throughput of Baseline-LAMB with 64 GPUs, which is 
  12.2K examples per second for Phase 1 and 4.6K examples per second for Phase 2. The improved convergence time 
  of \adasum  is a result of the 20\% improvement in algorithmic efficiency as shown in Table~\ref{bertalgo}.}
  \label{bertsys}
\end{table}
The algorithmic efficiency of \adasum is only useful if the algorithm can be
implemented without substantial degradation in system efficiency.
Table~\ref{bertsys} shows the speedup of \adasum-LAMB when compared to
Baseline-LAMB for an effective batch size of 64K. On our GPU cluster, the
Baseline-LAMB processes 12.2K examples per second during Phase 1 and 4.6K
examples per second during Phase 2. This reduction in throughput arises because
Phase 2 has more computation due to increases sequence length. The speedup
numbers in the table are relative to this baseline. For instance, at 256 GPUs,
the baseline scales to a speedup of 3.789 (a perfect scaling would
be 4). Just as a comparison with published NVIDIA numbers~\cite{nvidiabert},
the baseline finishes \bertlarge in 260 minutes which is slower than 236 minutes
reported by NVIDIA. This difference is due to the NVIDA cluster having a
slightly more performant DGX-2H configuration using a higher clock speed
compared to our DGX-2 cluster.

Though \adasum performs more computation during allreduce, the reduction in throughput for 64 GPUs is less than 2\% for Phase 1 and less than 1\% for Phase 2. As we increase the number of GPUs, the additional computation results in lower scaling efficiency for Phase 1. For instance, \adasum incurs roughly a 5 \% reduction (13\% reduction) in throughput when compared to the baseline for Phase 1 on 256 (512) GPUs. On the other hand, the throughput for Phase 2 with \adasum shows similar scalability as the baseline as we increase the number of GPUs, with \adasum being faster sometimes. This is due to the fact that Phase 2 does more computation. 

The reduction in system efficiency is more than compensated by the 20
\% reduction in algorithmic efficiency. As such, \adasum achieves
faster time to accuracy than the baseline. In particular, \adasum
completes \bertlarge in 214 minutes, which is faster than NVIDA
reported numbers on 256 GPUs~\cite{nvidiabert}. On 512 GPUs, \adasum
reaches the desired accuracy in 118 minutes. We observed some of the
overhead is due to CUDA aware MPI (openmpi + ucx) was not as fast as
NCCL.  We are in the process of porting \adasum's allreduce to NCCL as
a consequence.

\subsection{Case Study: LeNet-5}

LeNet-5 is a model for the MNIST dataset that is commonly used as a tutorial for
neural networks. This pioneering image recognition model is very small by
today's standards. While speeding up training of LeNet-5 is not a priority, a benefit of evaluating \adasum on it is that \emph{it is practical to
do exhaustive hyperparameter search}.
This allows us to thoroughly explore the convergence properties of
\adasum{} using the following experimental setup. First, we will find a very
aggressive learning rate schedule for sequential training that barely reaches a
baseline accuracy. Now, keeping the number of epochs fixed, we can compare
different distributed training configurations by how far off they are from the
sequential accuracy.

We used the \torch version of LeNet-5 found in Horovod's examples with a
momentum based SGD optimizer with a batch size of 32. While the original example
has a fixed learning rate schedule of 10 epochs, we were able to bring this down
to 2 epochs using a linear warmup and decay from zero to zero.

The maximum accuracy of LeNet-5 on MNIST is in the 99.3\%-99.4\% range. Using
99.3\% as our target accuracy, we used an Azure cluster of NC-series VMs with
333 NVIDIA K80 GPUs to optimize the hyperparameters of the training script.  The
following configuration reliably reaches the target accuracy in sequential
training in 2 epochs (while no reliable configuration was found for 1.75
epochs):

\begin{center}
  \small
\begin{tabular}{@{}lll@{}}
Epochs & Max LR & Warmup \% \\
\midrule
2 & 0.0328 & 17\% \\
\end{tabular}
\end{center}

We keep the number of epochs constant in the following experiments. Since MNIST
has 60000 images, one GPU will take 1875 steps per epoch, while 32 GPUs would
take only 58 steps per epoch.

\begin{figure}
  \includegraphics[width=\columnwidth]{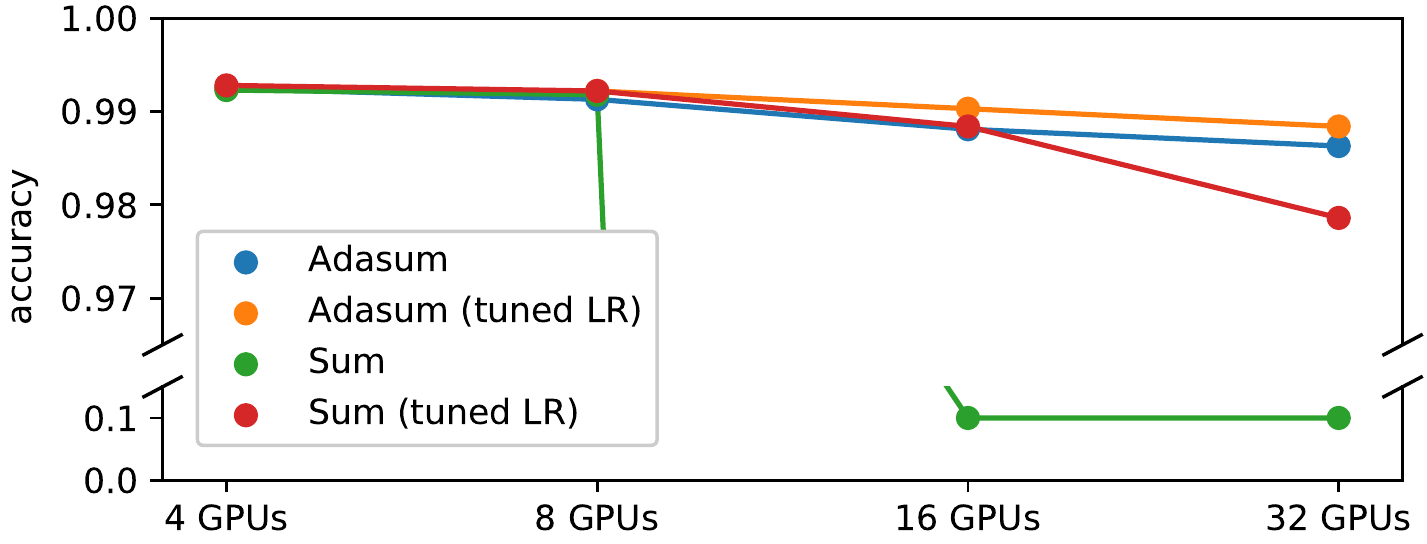}
  \caption{LeNet-5 accuracies under an aggressive sequential learning rate schedule for various distributed configurations. \label{fig:mnist}}
\end{figure}

We evaluated \adasum{} and Sum on 4, 8, 16 or 32 GPUs with both an unmodified
learning rate as well as an optimized one, which we searched for separately for
each combination of method and number of GPUs. Figure~\ref{fig:mnist} shows the
accuracies reached by each configuration under the aggressive learning rate
schedule for sequential training.

Without learning rate tuning Sum fails to converge at more than 8 GPUs, while
\adasum{} still converges at 32 GPUs \emph{without any hyperparameter search}.
This highlights the easy scalability that \adasum{} enables. Furthermore, even
with a tuned learning rate Sum is far below even untuned \adasum{} at 32 GPUs,
and is still beat by \adasum{} with a tuned learning rate at 16 GPUs.

Consider the tuned learning rates for each configuration:
\begin{center}
  \small
\begin{tabular}{@{}lllll@{}}
Method & 4 GPUs & 8 GPUs & 16 GPUs & 32 GPUs \\
\midrule
\adasum &   0.0328 & 0.0147 & 0.012 & 0.0204 \\
Sum & 0.0275 & 0.017 & 0.0089 & 0.0043 \\
\end{tabular}
\end{center}
For Sum going from 16 to 32 GPUs is coupled with a halving of the learning rate,
which means that the per iteration step size stays the same even though twice as
many GPUs are participating in each iteration. In contrast, \adasum{} can
maintain much higher learning rates at 16 and 32 GPUs.

\subsection{Case Study: Production Models}


The prior case studies demonstrate the efficacy of \adasum{} on public
models. Over the past three years, we have applied \adasum to a variety of production models from Company X\footnote{Name removed for anonymity.} 
and observed similar improvement in convergence with little or no hyperparameter tuning. 
Due to space constraints, we only provide a summary of these results to show the broad applicability of \adasum. 
For instance, the team relying on an LSTM-based model for predicting the next command a user is likely to issue in an application was able to use \adasum to train on four times the amount of data to result in 6\% improvement in downstream accuracy. Similarly, for a speech translation model with the Momentum optimizer, \adasum allowed the team to scale to 16 GPUs dramatically reducing the turn around time for their engineers. On a GPT-2 based model modified to work with code, \adasum provided 15\% faster time to accuracy. 


\section{Related Work}\label{sec:related}

Previous works for enabling large-batch training have focused on the problem of
adapting the learning rate appropriately. This is helpful because while large
batch sizes permit taking larger steps in many situations, the appropriate
learning rate changes during training.
Adam~\cite{adam} adjust the step size based on the variance of gradients, taking
smaller steps when variance is high.
LARS~\cite{lars} adapts learning rates per-layer using a \emph{trust ratio}
calculated as the ratio between the norm of the layer weights and the norm of
gradients, the intuition being that divergence happens when steps are large in
relation to the parameters being updated.
LAMB~\cite{lamb} can be seen as LARS applied to Adam instead of vanilla SGD.
These approaches that use statistical measures to adapt learning rate are
qualitatively different from \adasum, which exploits a specific property
(orthogonality) of gradients to take bigger steps when appropriate. \adasum and
learning rate adaption methods are in many cases complementary, as we have shown
in our experiments successfully combining \adasum with Adam and LAMB.

Asynchronous SGD~\cite{Dean2012,Chilimbi2014} approaches can address two issues in distributed synchronous
SGD: synchronization overhead of faster nodes having to wait for stragglers to
finish the iteration, and non-overlapping of compute and communication. However,
stale gradients present another potential source of degraded convergence. 
Specifically, the DC-ASGD algorithm~\cite{msra} addressed this staleness using an approximation 
of the Hessian as used in \adasum. 
They only use the diagonal elements of the $g \cdot g^T$ approximation of the Hessian and require  an additional hyperparameter which requires a careful tuning over time. It  was also only evaluated for SGD and Momentum-SGD. 
Our approach was motivated to be a drop-in replacement of the allreduce operation and thus we eliminate all hyperparameters in our combination and it is optimizer agnostic.
Similarly, Maleki et al.~\cite{ouripdps} use the Hessian to reduce staleness and use 
a Johnson-Lindenstrauss projection to get a low rank approximation of a semantics-preserving model combiner. 
Their approach only works with exact Hessian computation and is unlikely to scale to DNNs.

While large-batch training methods decrease the amount of communication needed,
\emph{gradient compression} approaches reduce the cost of each communication
round. In gradient quantization approaches gradients are cast to a lower
bit-width datatype for communication, with bit-widths ranging all the way down
to 1 bit~\cite{1BitSGD}. Low-rank compression methods communicate the most
important dimensions of gradients~\cite{PowerSGD}. Any lossy gradient
compression presents yet another potential source for loss of convergence.

\section{Conclusion}\label{sec:conclusion}
\adasum unlocks an unprecedented level of scalability for data
parallel distributed training of large models. Our case studies show
that \adasum scales \bertlarge-LAMB to 128k batch size, \bertlarge-Adam
to 64K batch size, and ResNet-50 Momentum to 64k batch size, while maintaining downstream accuracy.
\adasum is publicly available through open-source Horovod package~\cite{horovod}, can be simply enabled with an additional flag, and 
requires no additional hyper-parameters.


Finally, \adasum opens a new direction for future work on lightweight
ways to exploit orthogonality in model updates. We expect \adasum's
empirical observations to drive future work in adaptive learning rate
methods and further communication optimizations.




\bibliographystyle{plain}
\bibliography{refs}
  \appendix 
  \label{appendix}
  \section{Deriving sequential emulation}
  
Suppose $L_1(w)$ and $L_2(w)$ are two loss functions corresponding
to two different examples. Starting from model $w_0$, sequential
SGD, calculates $w_1=w_0 - \alpha \nabla L_1(w_0)$
followed by $w_2=w_1-\alpha \nabla L_2(w_1)$
where $\alpha$ is a properly set learning rate for both iteration. 
With forward substitution,
$w_2=w_0-\alpha (\nabla L_1(w_0) + \nabla L_2(w_1))$. 
Alternatively, $\nabla L_1(w_0)$ and 
$\nabla L_2 (w_0)$ (note that gradients are both
at $w_0$) are computed in parallel and $w$ is updated with $w_2'=w_0
-\alpha(\nabla L_1(w_0) + \nabla L_2(w_0))$. Clearly $w_2'$ and $w_2$
are different because $\nabla L_2$ was computed at a different point.

\subsection{Using Taylor Expansion}
\label{app:taylor}
\adasum uses an estimation for $\nabla L_2(w_1)$ using the Taylor expansion
to capture the effect of the first update on the second update. Note
that $\nabla L_2(w_1)$ is a convenient notation for the first
order derivative of $L_2$ and can be re-written with $\frac{\partial L_2}{\partial w}\Big|_{w_1}$.
Therefore:
\begin{equation}\label{eq:te}
\begin{split}
\nabla L_2(w_1)&=\frac{\partial L_2}{\partial w}\Big|_{w_1}=\frac{\partial L_2}{\partial w}\Big|_{w_0+(w_1-w_0)} = \frac{\partial L_2}{\partial w}\Big|_{w_0} \\
&+\frac{\partial^2 L_2}{(\partial w)^2}\Big|_{w_0}\cdot (w_1-w_0)  + O(\norm{w_1-w_0}^2) \\
&=\frac{\partial L_2}{\partial w}\Big|_{w_0}-\alpha \frac{\partial^2 L_2}{(\partial w)^2}\Big|_{w_0} \cdot \nabla L_1(w_0) \\
&+ \alpha^2 O(\norm{\nabla L_1(w_0)}^2)\\
&\approx \nabla L_2(w_0) - \alpha H_2(w_0) \cdot \nabla L_1(w_0)
\end{split}
\end{equation}
where the error is $\alpha^2O(\norm{\nabla L_1 (w_0)}^2)$ and $H_2(w_0)$ is the Hessian matrix of loss function $L_2$. 
The quadratic relationship between the error in Formula~\ref{eq:te} and the learning rate $\alpha$ helps
this approximation as in most training practices, learning rate decays with progress in training. However,
computing $H_2(w_0)$ requires significant computation powers as the size of this matrix is the number of model
parameters squared and with millions of parameters, even storing the matrix is infeasible. 

\cite{ggt} shows that for models with negative log likelihood loss functions (which is the case for
all models studied in this paper), the Hessian matrix can be approximated by the outer product of the gradients.
By using Equation~\ref{eq:te} and this approximation, $\nabla L_2(w_1)$ can be rewritten by:
\begin{equation} \label{eq:teggt}
\nabla L_2(w_1)\approx \nabla L_2(w_0) - \alpha \nabla L_2(w_0) \cdot \nabla L_2(w_0)^T \cdot \nabla L_1(w_0)
\end{equation}

\subsection{Choosing optimal learning rate}
\label{app:optimallr}
For this discussion, we assumed that $\alpha$ was properly chosen for the sequential
SGD algorithm. We use Taylor expansion for the loss function $L_2(w_0-\alpha \nabla L_2(w_0))$ and 
we will take its derivative with respect to $\alpha$ to find the optimal value:
\begin{equation}\label{eq:alpha}
\begin{split}
  &L_2(w_0-\alpha \nabla L_2(w_0)) \approx  L_2(w_0) - \alpha \nabla L_2(w_0)^T\cdot L_2(w_0) \\
  &                                + \frac{\alpha^2}{2} \nabla L_2(w_0)^T \cdot H_2(w_0) \cdot \nabla L_2(w_0)\\
  &\implies \frac{\partial L_2(w_0-\alpha \nabla L_2(w_0))}{\partial \alpha} = -\nabla L_2(w_0)^T\cdot L_2(w_0)\\
  & + \alpha \nabla L_2(w_0)^T \cdot H_2(w_0) \cdot \nabla L_2(w_0) = 0 \\
  &\implies \alpha \norm{\nabla L_2(w_0)}^4 = \norm{\nabla L_2(w_0)}^2 \implies \alpha = \frac{1}{\norm{\nabla L_2(w_0)}^2}
\end{split}
\end{equation}
where the last line is derived from the approximating for the Hessian matrix.
By putting together Equation~\ref{eq:alpha} and
\ref{eq:teggt}, $\nabla L_2(w_1)$ can be approximated by:
\begin{equation} \label{eq:adasumpre}
  \nabla L_2(w_1)\approx \nabla L_2(w_0) - \frac{\nabla L_2(w_0) \cdot \nabla L_2(w_0)^T}{\norm{\nabla L_2(w_0)}^2} \nabla L_1(w_0)
\end{equation}

Therefore, to approximate the sequential SGD semantics in a parallel setting, \adasum uses:
\begin{equation} \label{eq:adasum}
  \begin{split}
  w_2&=w_0-\alpha (\nabla L_1(w_0) + \nabla L_2(w_1)) \approx w_0 \\
  & - \alpha (\nabla L_1(w_0) + \nabla L_2(w_0) - \frac{\nabla L_2(w_0) \cdot \nabla L_2(w_0)^T}{\norm{\nabla L_2(w_0)}^2} \nabla L_1(w_0))\\
  &=w_0 - \alpha (g_1 + g_2 - \frac{g_2\cdot g_2^T}{\norm{g_2}^2}g_1)
\end{split}
\end{equation}
where in the last equality, $\nabla L_1(w_0)$ was replaced by $g_1$ and $\nabla L_2(w_0)$
by $g_2$ for simplicity (note that the gradients in the last equality are all from $w_0$). 
Equation~\ref{eq:asyncadasum} is derived from Equation~\ref{eq:adasum} by rearranging the
terms.

\subsection{Convergence Proof for \adasum}
\label{sec:conv}
\cite{conv} discusses the requirements for a training algorithm to converge to its optimal answer. Here we will present a simplified version of Theorem 1 and Corollary 1 from~\cite{conv}. 

Suppose that there are $N$ training examples for a model with loss functions $L_1(w),\dots,L_N(w)$ where $w$ is the model
parameter and $w_0$ is the initial model.  
Define $L(w)=\frac{1}{N}\sum_i L_i(w)$. Also assume that $w^*$ is the optimal model where $L(w^*)\leq L(w)$ for all $w$s. A training algorithm is {\em pseudogradient} if:
\begin{itemize}
	\item It is an iterative algorithm where $w_{i+1} = w_i-\alpha_i h_i$ where $h_i$ is a random vector and $\alpha_i$ is a scalar.
	\item $\forall \epsilon \exists \delta: E(h_i)^T\cdot \nabla L(w)\geq \delta>0$ where $L(w)\geq L(w^*)+\epsilon$ and $w^*$ is the optimal model.
	\item $E(\norm{h_i}^2)<C$ where $C$ is a constant.
	\item $\forall i: \alpha_i\geq 0$, $\sum_i \alpha_i=\inf$, and $\sum_i \alpha_i^2 < \inf$.
\end{itemize}

The following Theorem is taken from~\cite{conv}.
\begin{theorem}\label{conv-theorem}
A pseudogradient training algorithm converges to the optimal model $w^*$.
\end{theorem}

In this section, we assume that the true gradient, $\nabla L$ is bounded at any point.
As a reminder, $\adasum(g_1,g_2)=(1 - \frac{g_1 \cdot g_1^T}{2 \cdot \norm{g_1}^2}) \cdot g_2 + (1 - \frac{g_2 \cdot g_2^T}{2 \cdot \norm{g_2}^2}) \cdot g_1$.
As discussed in Section~\ref{sec:binarytree} \adasum operator reduces $N$ gradients in a binary tree manner.
We will prove that the final gradient has all necessary requirements of pseudogradient.
First, we discuss the inner product of \adasum final vector with $\nabla L(w)$:

\begin{lemma}\label{lem:angle}
  Suppose $X=\{x_1,\dots,x_N\}$ is a random variable distribution. For all $a$ and $b$ independently
  chosen from $X$, let's define $Y=\adasum(a,b)$. Assume that $\theta$ is the angle between $E(X)$ and $E(Y)$.
  $\cos\theta>0.942$.
\end{lemma}

\begin{proof}
  \begin{equation}\label{eq:expectation}
  \begin{split}
    &E(Y)=E\big(\adasum(a,b)\big)= E\Big((1 - \frac{a \cdot a^T}{2 \cdot \norm{a}^2}) \cdot b \\
    &+ (1 - \frac{b \cdot b^T}{2 \cdot \norm{b}^2}) \cdot a\Big) = E(a)+E(b)-E\Big(\frac{a \cdot a^T}{2 \cdot \norm{a}^2}\Big)\cdot E(b)\\
    &-E\Big(\frac{b \cdot b^T}{2 \cdot \norm{b}^2}\Big)\cdot E(a)=2E(X)-E\Big(\frac{a \cdot a^T}{\cdot \norm{a}^2}\Big)\cdot E(X)
  \end{split}
  \end{equation}
  where the last equation comes from the independence of $a$ and $b$. Next we will calculate,
   $\eta$, the angle between $2E(X)-\frac{a \cdot a^T}{\cdot \norm{a}^2}\cdot E(X)$ for some arbitrary $a$. 
   First let's denote $E(X)$ with $r$ and assume the angle between $r$ and $a$ is $\gamma$.
  By using the property of inner product, we have:
  \begin{equation}
    \begin{split}
      &\cos\eta=\frac{r^T\cdot(2r-\frac{a \cdot a^T}{\cdot \norm{a}^2}\cdot r)}{\norm{r}\cdot\norm{2r-\frac{a \cdot a^T}{\cdot \norm{a}^2}\cdot r}}\\
      &=\frac{2\norm{r}^2-\norm{r}^2(\cos\gamma)^2}{\norm{r}\cdot\sqrt{4\norm{r}^2+\norm{r}^2(\cos\gamma)^2-4\norm{r}^2(\cos\gamma)^2}}\\
      &=\frac{2-(\cos\gamma)^2}{\sqrt{4-3(\cos\gamma)^2}}\\
    \end{split}
    \end{equation}
    By taking a derivative of $\gamma$ from the last equation, we find the minimum value of 
    $\cos\eta$ to be $\approx 0.9428$ which concludes that $eta$ is at most $0.108\pi$. 
    Since in Formula~\ref{eq:expectation}, $E(Y)$ is calculated over an average of all possible
    $a$ vectors, we can still guarantee that $E(Y)$ and $E(X)$ have at most an angle of $0.108\pi$
    since we derived this value for the worst case scenario.
  \end{proof}

  \begin{lemma}\label{lem:norm}
    With same assumptions as in Lemma~\ref{lem:angle}, $\norm{E(X)}\leq\norm{E(Y)}\leq2\norm{E(X)}$.
  \end{lemma}
  \begin{proof}
    As discussed in Lemma~\ref{lem:angle}, $E(Y)=2E(X)-E(\frac{a \cdot a^T}{\cdot \norm{a}^2})\cdot E(X)
    =(2I-E(\frac{a \cdot a^T}{\cdot \norm{a}^2})\cdot E(X)$. It is trivial to check that the matrix
     $(2I-E(\frac{a \cdot a^T}{\cdot \norm{a}^2})$ is symmetric with eigenvalues between $1$ and $2$.
    Therefore, $\lambda_{\min}\norm{X}\leq\norm{Y}\leq\lambda_{\max}\norm{X}$ where $\lambda_{\min}$ 
    and $\lambda_{\max}$ are respectively the smallest and largest eigenvalues of the aforementioned matrix.
  \end{proof}

{\bf Assumption:} Lemma~\ref{lem:angle} showed that in the worst case  $E(Y)$ can rotate at most $0.108\pi$ 
with respect to $E(X)$. Even meeting this worst case requires carefully crafted $x_i$s. If \adasum
was applied recursively on $X$ ($Y=\adasum(X,X),Z=\adasum(Y,Y),\dots$) for $k$ times, the expected value 
of final distribution will at most have an angle of $0.108k\pi$ which is only possible if each $\adasum$
meets the worst case scenario and each worst case is stacked over the previous one. As one can imagine,
this is an extremely unlikely scenario.  
In case $x_i$s are gradients, we assume that $\adasum$ recursively always keeps the angle with $E(X)$ to at most
$\sigma$ where $\cos \sigma > 0$. Using this assumption and Lemma~\ref{lem:norm}, we can prove that $\adasum$ algorithm is a pseudogradient training algorithm.

\begin{theorem}
  $\adasum$ algorithm applied in an iterative manner using a proper learning rate on a set of $N$ gradients, 
  $G=\{g_1,\dots,g_N\}$ computed in parallel, is a pseudogradient training algorithm.
\end{theorem}
\begin{proof}
  Given that \adasum follows the iterative method of $SGD$, the first assumption of a pseudogradient training algorithm is
  met. Also, since we use the learning rate schedule from the converging sequential SGD, 
  the requirement for the learning rate is trivially met. Section~\ref{sec:binarytree} discussed how \adasum
  reduces all gradients in a binary tree manner which has $\log N$ steps. The distribution of 
  the leaf level in this binary tree is $G$ and the next level's distribution is $G_1=\adasum(G,G)$. Level $i$'s
  distribution is $G_{i+1}=\adasum(G_i,G_i)$. At the top of the tree, we have $G_{\log N}$ as the distribution. Using
  Lemma~\ref{lem:norm}, $\norm{E(G)}\leq\norm{E(G_{\log N})}\leq 2^{\log N} \norm{E(G)}=N \norm{E(G)}$. Since $E(G)$
  is the true gradient ($\nabla L(w_i)$), $\norm{E(G)}$ is bounded by assumption and therefore, so is $\norm{E(G_{\log N})}$.
  This meets the requirement for the norm of the $h_i$s. Finally, Lemma~\ref{lem:norm} proves that $E(G)\leq \norm{E(G_{\log N})}$ and therefore  $E(G_{\log N})^T\dot E(G)\geq \norm{E(G_{\log N})}\norm{E(G)}\cos\sigma\geq \norm{E(G)}^2\cos\sigma$.
  For any $w_i$ which is not $w^*$, $\norm{E(G)}>0$ and based on the assumption, $\cos \sigma>0$. Therefore, the positive
  inner product assumption is also met which concludes that \adasum is a pseudogradient training algorithm and it converges.
\end{proof}

\subsection{\adasum Convergence Rate}
Convergence rate of \adasum is highly dependent on orthogonality of the gradients. In the worst case scenario if
all gradients are parallel, the algorithm converges in $1/N$ rate of the sequential SGD where $N$ is the number of processors and in the best case where all of the gradients are orthogonal, we expect \adasum to converge
as fast the sequential SGD. 


\subsection{Convergence of \adasum with Learning Rate Optimizers}
The kernel of the computation in all optimizers such as Adam, LAMB or LARS is 
computing the gradients. These optimizers differ by having different learning rate
mechanisms for each parameter. Generally one can think of having a dynamic
learning rate per layer for each of these optimizers. At each iteration $i$ 
and for each layer $l$, let's define $\Delta_i^l=w_{i+1}^l-w_{i}^l$
to be the update term for layer $l$. If $g_i^l$ is the corresponding 
gradient at iteration $i$, $\Delta_i^l\approx C_i^l \cdot g_i^l$ where $C_i^l$
is a scalar. This approximation is true in expectation for all optimizers
that use an unbiased gradient.

\adasum reduces all $\Delta_i^l$ across all GPUs for each layer and since
$\Delta_i^l$ is approximately a constant multiplied by the gradient,
\adasum applied on $\Delta_i^l$s is the same as applying it on $g_i^l$ scaled
by the same constant.

\end{document}